\documentclass[preprint,11pt,letterpaper]{article}

\usepackage[T1]{fontenc}
\usepackage{amsthm,amsmath,cite}
\usepackage{amssymb}\usepackage{dsfont}\usepackage{array}\usepackage{mathrsfs}\usepackage{graphicx}
\RequirePackage[colorlinks,citecolor=blue,urlcolor=blue]{hyperref}
\newtheorem{theorem}{Theorem}[section]
\newtheorem{lemma}{Lemma}[section]
\newtheorem{corollary}{Corollary}[section]

\newtheorem{proposition}{Proposition}[section]
\newcommand{\nn}{\mathcal{N}}

\newcommand{\g}{\mathcal{G}}
\newcommand{\np}{n_+}
\newcommand{\nm}{n_-}
\newcommand{\npm}{n_{\pm}}
\newcommand{\PP}{{\mathbf{P}}}
\newcommand{\FF}{{\mathcal{F}}}
\newcommand{\GG}{{\mathcal{G}}}
\newcommand{\EE}{{\mathbf{E}}}
\newcommand{\II}{\mathbf{1}}
\newcommand{\RR}{\mathbb{R}}

\begin{document}
\bibliographystyle{IEEEtran}
\title{Community detection thresholds\\and the weak Ramanujan property}


\author{Laurent Massouli\'e\\
Microsoft Research--Inria Joint Centre\\
laurent.massoulie@inria.fr}
\date{\today}
\maketitle


\begin{abstract}
The present work is concerned with community detection, that is reconstruction of hidden components, in random graph models. 
Decelle et al.~\cite{Decelle11} conjectured the existence of a sharp threshold on model parameters for community detection in sparse random graphs drawn from the {\em stochastic block model}. Mossel, Neeman and Sly~\cite{Mossel12} established the negative part of the conjecture, proving impossibility of meaningful reconstruction below the threshold. The positive part of the conjecture remained elusive so far: results of Coja-Oghlan~\cite{Coja-oghlan10} imply that a particular spectral method applied to the graph's adjacency matrix achieves non-trivial reconstruction, but this applies only when above the conjectured threshold by a possibly large constant.

In this work we solve the positive part of the conjecture. To that end we introduce a modified adjacency matrix $B$ based on neighborhood expansion. Specifically $B$ counts {\em simple}, or {\em self-avoiding} paths of a given length $\ell$ between pairs of nodes. We then prove that for logarithmic  length $\ell$, the leading eigenvectors of this modified matrix provide a non-trivial reconstruction of the underlying structure, thereby settling the conjecture. A key step in the proof consists in establishing a {\em weak Ramanujan property} of the constructed matrix $B$. Namely, the spectrum of $B$ consists in two leading eigenvalues $\rho(B)$, $\lambda_2$ and $n-2$ eigenvalues of a lower order $O(n^{\epsilon}\sqrt{\rho(B)})$ for all $\epsilon>0$, $\rho(B)$ denoting $B$'s spectral radius. 

Graphs are called Ramanujan when they are $d$-regular with second eigenvalue $\lambda\le 2 \sqrt{d-1}$. Random $d$-regular graphs were shown to have a second largest eigenvalue $\lambda$ of $2\sqrt{d-1}+o(1)$ by Friedman~\cite{friedman08}, thus being {\em almost} Ramanujan. Erd\H{o}s-R\'enyi graphs with average degree $d$ at least logarithmic ($d=\Omega(\log n)$) were shown by Feige and Ofek~\cite{Feige05} to have a second eigenvalue of $O(\sqrt{d})$, a slightly weaker version of the Ramanujan property. 

However this spectrum separation property fails for sparse ($d=O(1)$)  Erd\H{o}s-R\'enyi graphs, whose spectrum is dominated by the presence of high-degree nodes. Our result thus shows that by constructing matrix $B$ through neighborhood expansion, we regularize the original adjacency matrix to eventually recover a weak form of the Ramanujan property.
\end{abstract}
\thispagestyle{empty}
\newpage
\setcounter{page}{1}



\section{Introduction}
\subsection{Background}
Community detection, like clustering, aims to identify groups of similar items from a global population. It is a useful primitive for performing recommendation, e.g. of contents or contacts to users of online social networks. The stochastic block model has been introduced by Holland et al.~\cite{Holland83} to represent interactions between individuals. It consists of a random graph on $n$ nodes, each node $i\in\nn=\{1,\ldots,n\}$ being assigned a type $\sigma_i$ from some fixed set $\Sigma$. Conditionally on node types, edge $(i,j)$ is present with probability $p(\sigma_i,\sigma_j)$ independently of other edges, for some matrix of probabilities $(p(\sigma,\sigma'))$.

It constitutes an adequate testbed for community detection. Indeed the performance of candidate detection schemes, captured by the fraction of nodes $i$ for which estimated types $\hat{\sigma}_i$ and true types $\sigma_i$ coincide, can be compared and analysed on instances of the stochastic block model. Such analyses can in turn suggest new schemes.

Recently Decelle et al.~\cite{Decelle11} conjectured the existence of a phase transition in the sparse regime where the graph's average degree is $O(1)$. Specifically, they predicted that for parameters below a certain threshold, no estimates $\hat{\sigma}_i$ of node types existed that would be positively correlated with true types $\sigma_i$, while above the threshold, belief propagation algorithms could determine estimates $\hat{\sigma}_i$ achieving such a positive correlation. Their conjecture is formulated on a simple symmetric instance of the stochastic block model featuring two node types $\{+1,-1\}$. The phenomenon appears more general though: Heimlicher et al.~\cite{NIPS12} extended the conjecture to the more general setup of labeled stochastic block models.

The study of this phenomenon is important for two reasons. First, by localizing precisely the transition point below which no useful signal is present in the observations, one thus characterizes how much subsampling of the original graph can be performed before all information is lost. Second, algorithms leading to estimates $\hat{\sigma}_i$ that achieve positive correlation all the way down to the transition are expected to constitute more robust approaches than alternatives which would fail before the transition. It is therefore important to determine such algorithms.

The negative part of the conjecture has been proven by Mossel, Neeman and Sly~\cite{Mossel12}. Essentially they established that existence of estimates $\hat{\sigma}_i$ positively correlated with true types $\sigma_i$ would imply feasibility of a reconstruction problem on a random tree model describing the local statistics of the original random graph. However by results of Evans et al.~\cite{evans00} such reconstruction is infeasible below the conjectured transition point.

Until now, positive results in the sparse case did not apply down to the transition point. The best results to date (see \cite{Mossel12}) relied on Coja-Oghlan~\cite{Coja-oghlan10}, showing that spectral clustering applied to the adjacency matrix, suitably trimmed by removal of high degree nodes, yields positively correlated estimates. However this does not apply down to the conjectured threshold. 

This limitation stems from the following fact. Spectral methods perform well on matrices enjoying a spectral separation property, namely the spectrum should comprise a few large eigenvalues whose associated eigenvectors reflect the sought structure and all other eigenvalues should be negligible. The prototype of such separation is the Ramanujan property, according to which $d$-regular graphs have the second eigenvalue $\lambda$ no larger than $2\sqrt{d-1}$ in absolute value. Friedman~\cite{friedman08} established that random $d$-regular graphs almost satisfy this, in that for them $|\lambda|\le 2\sqrt{d-1}+o(1)$. Erd\H{o}s-R\'enyi graphs with average degree $d$ are such that $|\lambda|\le O(\sqrt{d})$, {\em provided} $d=\Omega(\log n)$ (see Feige and Ofek~\cite{Feige05}), but such Ramanujan-like separation is lost for smaller $d$. This lack of separation inherently limits the power of spectral methods in the sparse case.

\subsection{Main results}
We focus on the stochastic block model in Decelle et al.~\cite{Decelle11}. The graph is denoted $\g$, node types (or spins) $\sigma_i$ are uniformly and i.i.d. drawn from $\{-1,+1\}$. An edge is present between any two nodes $i$,$j$ with probability $a/n$ if $\sigma_i=\sigma_j$, and $b/n$ if $\sigma_i=-\sigma_j$, constants $a$ and $b$ being the model parameters. The conjectured transition point is specified by quantity $\tau=(a-b)^2/[2(a+b)]$: for $\tau<1$ it is known that positively correlated detection is impossible; we set out to prove that it is feasible for $\tau>1$. 

We shall make use of the notations $\alpha:=(a+b)/2$, $\beta:=(a-b)/2$. The detectability condition $\tau>1$ can be  restated as 
\begin{equation}\label{c1} 
\beta^2> \alpha.
\end{equation}
As mentioned, Coja-Oghlan regularizes the adjacency matrix of the random graph by removing high degree nodes before applying spectral clustering. In contrast, we regularize the initial data through {\em path expansion}.
Namely, we do not work directly with the adjacency matrix of the orginal graph. Instead we form matrix $B^{(\ell)}$, where $B^{(\ell)}_{ij}$ counts the number of self-avoiding paths of graph edges of  length $\ell$ connecting $i$ to $j$. 

Our main result is then the following
\begin{theorem}\label{th_0}Assume Condition~(\ref{c1}) holds. 
Set the path length parameter $\ell$  such that $\ell\sim c\log(n)$ for a constant $c$ such that $c\log(\alpha)<1/4$. Let $x$ be a normed eigenvector corresponding to the second largest eigenvalue of $B^{(\ell)}$. There exists $t$ such that, defining the spin estimates $\hat{\sigma}_i$ as
\begin{equation}
\hat{\sigma}_i=\left\{\begin{array}{ll}
+1&\hbox{if } x_i\sqrt{n}\ge t,\\
-1&\hbox{otherwise,}
\end{array}
\right.
\end{equation}
the empirical {\em overlap} between the true and estimated spins defined as
\begin{equation}\label{hihi} 
\hbox{ov}(\sigma,\hat{\sigma}):=\frac{1}{n}\sum_{i\in \nn}\sigma_i\hat{\sigma}_i
\end{equation} 
converges in probability to the set $\{-r,+r\}$ for some strictly positive constant $r>0$ as $n\to\infty$.
\end{theorem}
It proves the positive part of Decelle et al.'s conjecture. It further identifies a specific spectral method based on the path-expanded matrix $B^{(\ell)}$. An auxiliary result consists in showing that matrix $B^{(\ell)}$ enjoys a spectral separation property, that is a weak version of the Ramanujan property. Namely, denoting by $\rho(B^{(\ell)})$ the spectral radius of $B^{(\ell)}$, we show that the third largest eigenvalue $\lambda$ of matrix $B^{(\ell)}$ satisfies for all positive constant $\epsilon$:
$$
|\lambda|\le n^{\epsilon}\sqrt{\rho(B^{(\ell)})}.
$$
We note that computation of $B^{(\ell)}$ and hence of the $\hat{\sigma}_i$ can be done in polynomial time: as shown in Lemma~\ref{cycles} the $\ell$-neighborhood of any $i$ contains at most one cycle so that each $B^{(\ell)}_{ij}$ is readily evaluated by suitable breadth-first search.
\subsection{Paper organization} 
Section~\ref{sec:2} contains the intermediate results involved in the proof and how they combine to establish Theorem~\ref{th_0}. Section~\ref{sec:3} proves Theorem~\ref{global}, which expresses matrix $B^{(\ell)}$ as an expansion in terms of the matrices $B^{(m)}$, $m<\ell$, together with bounds on the spectral norm of the matrix coefficients involved.
Section~\ref{sec:4} contains the so-called ``local analysis'' of node neighborhoods. Specifically it gives controls on the vectors $B^{(m)}e$ and $B^{(m)}\sigma$, where $e$ is the all-ones vector and $\sigma$ is the vector of spins, establishing a quasi-deterministic growth pattern with respect to $m$. Section~\ref{sec:conclusions} concludes. 
\section{Proof structure}\label{sec:2}
Our key objective is to determine the spectral structure of $B^{(\ell)}$. Specifically we wish to establish
\begin{theorem}\label{th_1} 
Assume (\ref{c1}) and $\ell=c\log n$ with $c\log(\alpha)<1/4$. 

(i) The leading eigenvalue of $B^{(\ell)}$ is up to logarithmic factors $\Theta(\alpha^{\ell})$, with corresponding eigenvector asymptotically parallel to $B^{(\ell)}e$. 

(ii) Its second eigenvalue is $\Omega(\beta^{\ell})$ up to logarithmic factors, with corresponding eigenvector asymptotically parallel to $B^{(\ell)}\sigma$. 

(iii) There is a random variable $X$ with unit mean and variance $1/(\beta^2/\alpha-1)$ such that for all $x$ that is an atom of neither $X$'s nor $-X$'s distribution, the following convergence in probability holds for any normed vector $y$ asymptotically aligned with $B^{(\ell)}\sigma$:
\begin{equation}\label{hoho}
\frac{1}{n}\sum_{i=1}^{n}\II\left\{y_i\ge \frac{x}{\sqrt{n \EE(X^2)}}\right\}\to \frac{1}{2}\left[\PP(X\ge x)-\PP(-X\ge x)\right].
\end{equation}

(iv) For any $\epsilon>0$, all other eigenvalues are of order $n^{\epsilon}\sqrt{\alpha^{\ell}}$.
\end{theorem} 
Before we describe the steps used to establish this, let us verify how it implies Theorem~\ref{th_0}. Note that since $\EE(X)=1$, writing
$$
\EE(X)=\int_0^{\infty}\left(\PP(X\ge x)-\PP(-X\ge x)\right)dx,
$$
we see that inequality $\PP(X\ge x)-\PP(-X\ge x)>0$ must hold on a set of $x$'s of positive Lebesgue measure. Since the points $x$ at which the distribution of either $X$ or $-X$ has an atom is at most countable, there thus exists an $x$ at which neither distribution has an atom, and the desired inequality $\PP(X\ge x)-\PP(-X\ge x)>0$ holds. Letting $t=x/\sqrt{\EE(X^2)}$ and $r=\PP(X\ge x)-\PP(-X\ge x)$ we readily have by (\ref{hoho}) that the empirical overlap in (\ref{hihi}) must converge to $\{-r,+r\}$.

Theorem~\ref{th_1} will follow from the combination of two analyses. Let $\bar{A}$ denote the expectation of the graph's adjacency matrix conditional on the spin vector $\sigma$, that is
\begin{equation}\label{bara}
\bar{A}=\frac{a}{n}\left[\frac{1}{2}(e e'+\sigma \sigma')-I\right]+\frac{b}{2n}(e e'-\sigma\sigma').
\end{equation}
The first analysis establishes the following
\begin{theorem}\label{global}
Matrix $B^{(\ell)}$ verifies the identity
\begin{equation}\label{matrix_expansion} 
B^{(\ell)}=\Delta^{(\ell)}+\sum_{m=1}^{\ell}(\Delta^{(\ell-m)}\bar{A} B^{(m-1)})-\sum_{m=1}^{\ell}\Gamma^{\ell,m},
\end{equation}
for matrices $\Delta^{(\ell)}$, $\Gamma^{\ell,m}$ such that for $\ell=O(\log n)$ and any fixed $\epsilon>0$, with high probability
\begin{eqnarray}\label{rho_delta} 
\rho(\Delta^{(\ell)})\le n^{\epsilon} \alpha^{\ell/2},\\
\label{rho_gamma} 
\rho(\Gamma^{\ell,m})\le n^{\epsilon-1}\alpha^{(\ell+m)/2},\;m=1,\ldots,\ell.
\end{eqnarray} 
\end{theorem}
A {\em local} analysis is then needed to establish properties of the $\ell$-neighborhoods of nodes in graph $\g$. 
The key quantities in this analysis are the following
\begin{equation}
\begin{array}{ll}
S_t(i)&=|\{j: d_{\g}(i,j)=t\}|,\\
D_t(i)&=\sum_{j}\II_{d_{\g}(i,j)=t}\sigma_j.
\end{array}
\end{equation}
They are close (in a sense made precise in Section~\ref{sec:4}) to the corresponding quantities $(B^{(t)}e)_i$, $(B^{(t)}\sigma)_i$, and are easier to analyze. In particular, they enjoy a {\em quasi-deterministic growth} property:
\begin{theorem}\label{lemma_LD_bounds}
Under Condition~(\ref{c1}) for some suitable constants $C$ and $\epsilon$, with probability $1-O(n^{-\epsilon})$ the following holds for all $i\in \nn$ and $\ell=O(\log n)$:
\begin{equation}\label{control_1}
\begin{array}{ll}
S_t(i)&\le C\log(n) \alpha^t,\; t=1,\ldots,\ell,\\
|D_t(i)|&\le C\log(n)\beta^t,\; t=1,\ldots,\ell.
\end{array}
\end{equation}
\begin{equation}\label{control_2}
\begin{array}{ll}
S_t(i)&= \alpha^{t-\ell}S_{\ell}(i)+O(\log(n)+\sqrt{\log(n)\alpha^{t}}),\; t=1,\ldots,\ell,\\
D_t(i)&=\beta^{t-\ell}D_{\ell}(i)+ O(\log(n)+\sqrt{\log(n)\alpha^{t}}),\; t=1,\ldots,\ell.
\end{array}
\end{equation}
\end{theorem}
This, combined with Theorem~\ref{global}, yields the key intermediate step:
\begin{theorem}\label{th_ramanujan} 
Let $\ell=c\log(n)$, where constant $c$ is such that $c\log(\alpha)<1/4$. Under Condition~(\ref{c1}) the matrix $B^{(\ell)}$ counting simple paths satisfies the following weak Ramanujan property
\begin{equation}
\sup_{|x|=1,x'B^{(\ell)}e=x'B^{(\ell)}\sigma=0}|B^{(\ell)}x|\le n^{\epsilon}\alpha^{\ell/2}.
\end{equation}
\end{theorem} 
Another key ingredient consists in coupling the neighborhoods of nodes in graph $\g$ with a random tree process, and performing a martingale analysis of this tree process. This is done in Section~\ref{sec:coupling}. It establishes (see Theorem~\ref{thm_local}) that the vector $(\beta^{-\ell}D_{\ell}(i))$ is close in some sense to a vector $(\sigma_i D_i)$ where the $D_i$ are i.i.d., distributed as the limit of a martingale. This limiting martingale distribution is precisely that of variable $X$ in the statement of Theorem~\ref{th_1}. 
\section{Matrix expansion and spectral radii bounds}\label{sec:3} 
Our aim in this section is to establish Theorem~\ref{global}. 
Denoting $\xi_{ij}$ the indicator of edge $(i,j)$'s presence in $\g$ we can write
\begin{equation}
B^{(\ell)}_{ij}=\sum_{i_0,i_1,\ldots,i_{\ell}\in\nn}\II_{i_0=i}\II_{i_{\ell}=j}\II_{|\{i_0,\ldots,i_{\ell}\}|=\ell+1}\prod_{t=1}^{\ell}\xi_{i_{t-1}i_t}.
\end{equation}
Denote by $P_{ij}$ the set of all so-called {\em self-avoiding}, or {\em simple} paths $i_0^{\ell}:=\{i_0,\ldots,i_{\ell}\}$ from $i$ to $j$ appearing in the above sum. Let
\begin{equation}
\Delta^{(\ell)}_{ij}:=\sum_{i_0^{\ell}\in P_{ij}}\prod_{t=1}^{\ell}(A-\bar{A})_{i_{t-1}i_t}
\end{equation}
where $\bar{A}$ is as in~(\ref{bara}).
We then have the expansion:
\begin{equation}\label{exp1} 
\Delta^{(\ell)}_{ij}=B^{(\ell)}_{ij}-\sum_{m=1}^{\ell}\sum_{i_0^{\ell}\in P_{ij}} \prod_{t=1}^{\ell-m}(A-\bar{A})_{i_{t-1}i_t} \bar{A}_{i_{\ell-m}i_{\ell-m+1}} \prod_{t=\ell-m+2}^{\ell} A_{i_{t-1}i_t}. 
\end{equation}
Introduce the set $Q^m_{ij}$ of paths $i_0^{\ell}$ defined by:
$$
i_0^{\ell}\in Q^m_{ij}\Leftrightarrow i_0=i\;\&\; i_{\ell}=j\;\&\; |\{i_0,\ldots,i_{\ell-m}\}|=\ell-m+1\;\&\; |\{i_{\ell-m+1},\ldots,i_{\ell}\}|=m.
$$
Paths in $Q^m_{ij}$ are thus concatenations of simple paths $i_0^{\ell-m}$ and $i_{\ell-m+1}^{\ell}$. Note that $P_{ij}\subset Q^m_{ij}$. Let $R^m_{ij}$ denote the set difference $Q^m_{ij}\setminus P_{ij}$. It then consists of paths $i_0^{\ell}$ such that both $i_0^{\ell-m}$ and $i_{\ell-m+1}^{\ell}$ are simple, and further verify that the intersection of the corresponding sets is not empty.

Define matrix $\Gamma^{\ell,m}$ as
\begin{equation}
\Gamma^{\ell,m}_{ij}:=\sum_{i_0^{\ell}\in R^m_{ij}} \prod_{t=1}^{\ell-m}(A-\bar{A})_{i_{t-1}i_t} \bar{A}_{i_{\ell-m}i_{\ell-m+1}} \prod_{t=\ell-m+2}^{\ell} A_{i_{t-1}i_t}.
\end{equation}
With these notations at hand, one obtains from~(\ref{exp1}):
\begin{equation}
\begin{array}{lll}
\Delta^{(\ell)}_{ij}&=&B^{(\ell)}_{ij}-\sum_{m=1}^{\ell}\sum_{i_0^{\ell}\in Q^m_{ij}} \prod_{t=1}^{\ell-m}(A-\bar{A})_{i_{t-1}i_t} \bar{A}_{i_{\ell-m}i_{\ell-m+1}} \prod_{t=\ell-m+2}^{\ell} A_{i_{t-1}i_t}\\
&&+\sum_{m=1}^{\ell}\sum_{i_0^{\ell}\in R^m_{ij}} \prod_{t=1}^{\ell-m}(A-\bar{A})_{i_{t-1}i_t} \bar{A}_{i_{\ell-m}i_{\ell-m+1}} \prod_{t=\ell-m+2}^{\ell} A_{i_{t-1}i_t}\\
&=&B^{(\ell)}_{ij}-\sum_{m=1}^{\ell}(\Delta^{(\ell-m)}\bar{A} B^{(m-1)})_{ij}+\sum_{m=1}^{\ell}\Gamma^{\ell,m}_{ij},
\end{array}
\end{equation}
where we noticed that the summation over paths in $Q^m_{ij}$ of the corresponding products yields the $(ij)$-entry of the product matrix $\Delta^{(\ell-m)}\bar{A}B^{(m-1)}$. This is precisely expansion~(\ref{matrix_expansion}).

We then have the following
\begin{proposition}\label{trace_delta}
For all integers $k,\ell\ge 1$, it holds that
\begin{equation}\label{bound_delta}
\EE\left[\rho(\Delta^{(\ell)})^{2k}\right]\le \sum_{v=\ell+1}^{k\ell+1}\sum_{e=v-1}^{k\ell} n^{v} [(v+1)^2(\ell+1)]^{2k(1+e-v+1)}\left(\frac{\alpha}{n}\right)^{v-1} \left[\frac{\max(a,b)}{n}\right]^{e-v+1}\cdot 
\end{equation}
\end{proposition}
The proof uses the trace method, bounding $\rho(\Delta^{(\ell)})^{2k}$ by the trace of $(\Delta^{(\ell)})^{2k}$ and a specific encoding of circuits involved in this trace computation. The fact that circuits consist in concatenations of simple paths of length $\ell$ is then leveraged to control which codes of circuits have to be considered in the trace bound. The details are provided in the Appendix.

Inequality~(\ref{rho_delta}) readily follows from Proposition~\ref{trace_delta}. Indeed for $\ell=O(\log(n))$ and fixed $\epsilon>0$, choose an integer $k>0$ such that $\epsilon>1/(2k)$.
By~(\ref{bound_delta}), it holds that
$$
\EE(\rho(\Delta^{(\ell)})^{2k})\le (1+o(1))n \alpha^{k\ell}[(k\ell+2)^2(\ell+1)]^{2k}.
$$
Thus
$$
\begin{array}{ll}
\PP(\rho\ge n^{\epsilon}\alpha^{\ell/2})&\le \frac{\EE(\rho^{2k})}{n^{2k\epsilon}\alpha^{k\ell}}\\
&\le (1+o(1) \frac{n\alpha^{k\ell}[(k\ell+2)^2(\ell+1)]^{2k}}{n^{2k\epsilon}\alpha^{k\ell}}\\
&\le (1-o(1))n^{1-2k\epsilon}[(k\ell+2)^2(\ell+1)]^{2k}\\
&=o(1),
\end{array}
$$
since we chose $k$ so that $2k\epsilon>1$ and the last term is polylogarithmic in $n$. This establishes~(\ref{rho_delta}).

We now establish a bound on the spectral radius of the matrix $\Gamma^{\ell,m}$ previously introduced. Specifically, we have
\begin{proposition}\label{trace_gamma}
For all $k,\ell\ge 1$ and $m\in\{1,\ldots,\ell\}$ we have the following
\begin{equation}\label{bound_gamma} 
\EE((\rho(\Gamma^{\ell,m})^{2k})\le\sum_{v=m\vee(\ell-m+1)}^{1+k(\ell+m)} 
\sum_{e=v-1}^{k(\ell+m)}\left(\frac{a\vee b}{n}\right)^{2k+e-v+1} v^{2k}[(v+1)^2(\ell+1)]^{4k(1+e-v+1)}
n^v \left(\frac{\alpha}{n}\right)^{v-1}.
\end{equation}
\end{proposition}
The proof, postponed to the appendix, again uses the trace method, a specific encoding of circuits involved in the trace bound, and then leverages the constraints on circuits appearing in this bound. It readily implies inequality~(\ref{rho_gamma}).  
Indeed for $\ell=O(\log(n))$, and any fixed $\epsilon>0$, choose $k>0$ such that $\epsilon>1/(2k)$.
By (\ref{bound_gamma}) it holds that
$$
\EE(\rho(\Gamma^{\ell,m})^{2k})\le (1+o(1))n \alpha^{k(\ell+m)}\left[\frac{a\vee b}{n} (k(\ell+m)+2)^5 (\ell+1)^2\right]^{2k}.
$$
Thus
$$
\begin{array}{ll}
\PP(\rho\ge n^{\epsilon-1}\alpha^{(\ell+m)/2})&\le \frac{\EE(\rho(\Gamma^{\ell,m})^{2k})}{n^{2k(\epsilon-1)}\alpha^{k(\ell+m)}}\\
&\le (1+o(1)) n^{1-2k\epsilon}\left[(a\vee b) (k(\ell+m)+2)^5 (\ell+1)^2\right]^{2k},
\end{array}
$$
and this last bound decays to zero as a power of $n$ by the condition $2k\epsilon>1$ and the fact that the last term in the product is polylogarithmic in $n$. This completes the proof of Theorem~\ref{global}. 
%
\section{Local Analysis: structure of expanded neighborhoods}
\label{sec:4}
This section is devoted to the structure of the local neighborhoods of nodes. 
We start with general bounds. We then relate vectors of interest $B^{(\ell)}e$ and $B^{(\ell)}\sigma$ to the neighborhood structures. The martingale analysis of neighborhood structures follows.
\subsection{Preliminaries}\label{sec:general_bounds} \
For any $k\ge 0$, the number of nodes with spin $\pm$ at distance $k$ (respectively $\le k$) of node $i$ is denoted $U^{\pm}_k(i)$ (respectively, $U^{\pm}_{\le k}(i)$). We thus have
\begin{equation}
S_t(i)=U^+_t(i)+U^-_t(i),\; D_t(i)=U^+_t(i)-U^-_t(i).
\end{equation}
We shall omit indices $i$ when considering quantities related to a fixed node $i$. 
In the remainder of the section we condition on the spins $\sigma$ of all nodes. We denote $\npm$ as the number of nodes with spin $\pm$.

For fixed $i\in \nn$ it is readily seen that,  conditionally on ${\mathcal F}_{k-1}:=\sigma(U^+_t,U^-_t,t\le k-1)$, we have:
\begin{equation}\label{dynamics1} 
\begin{array}{ll}
U^+_k&\sim \hbox{Bin}\left(\np-U^+_{\le k-1}, 1-(1-a/n)^{U^+_{k-1}}(1-b/n)^{U^-_{k-1}}\right),\\
U^-_k&\sim \hbox{Bin}\left(\nm-U^-_{\le k-1}, 1-(1-a/n)^{U^-_{k-1}}(1-b/n)^{U^+_{k-1}}\right).
\end{array}
\end{equation}
Theorem~(\ref{lemma_LD_bounds}) is established based on these characterizations by extensive use of Chernoff bounds for binomial variables. Its proof is deferred to the Appendix. 

The next technical result establishes approximate independence of neighborhoods of distinct nodes. It is instrumental in Section~\ref{sec:coupling} e.g. in establishing weak laws of large numbers on the fraction of nodes satisfying a given property. 
\begin{lemma}\label{coupling_lemma_2}
Consider any two fixed nodes $i$, $j$ with $i\ne j$. Let $\ell=c\log(n)$ where constant $c$ is such that $c\log(\alpha)<1/2$. 
Then the variation distance between the joint law of their neighborhood processes ${\mathcal{L}}((U^{\pm}_k(i))_{k\le \ell},(U^{\pm}_k(j))_{k\le \ell})$ and the law with the same marginals and independence between them, denoted ${\mathcal{L}}((U^{\pm}_k(i))_{k\le \ell})\otimes {\mathcal{L}}((U^{\pm}_k(j))_{k\le \ell})$,  goes to zero as a negative power of $n$ as  $n\to\infty$. 
\end{lemma}
\begin{proof}
Take two independent realizations of the processes $(U^{\pm}_k(i))_{k\le \ell})$ and $(U^{\pm}_k(j))_{k\le \ell})$. 
Use them to perform a joint construction of the two processes as follows. Having constructed the sets ${\mathcal{U}}^{\pm}_t(i)\subset \nn$, ${\mathcal{U}}^{\pm}_t(j)\subset \nn$ for $t=1,\ldots,k-1$ and assuming the $i$-sets and the $j$-sets have not yet met, we construct them at step $k$ as follows. To construct $\mathcal{U}_k^{\pm}(i)$ we select a size $U^{\pm}_k(i)$ subset uniformly at random from $\nn^{\pm}\setminus {\mathcal{U}}^{\pm}_{\le k-1}(i)$. We do similarly for $j$. The construction can proceed based on the independent inputs so long as the resulting $i$-sets and $j$-sets do not intersect. However on $\cap_{t\le k}\{S_t(i)\vee S_t(j)\le C\log(n)\}$, the expected size of the intersection will be upper-bounded by $O(\log^2(n) \alpha^{2k}/n)=O(\log^2(n) n^{-2\epsilon}$, where $c\log(\alpha)=1/2-\epsilon$. The controls in the proof of Theorem~\ref{lemma_LD_bounds} ensure that the probability of $\cap_{t\le k}\{S_t(i)\vee S_t(j)\le C\log(n)\}$ is $1-O(n^{-\epsilon})$ and the result follows. 
\end{proof}
We now state a lemma on the presence of cycles in the $\ell$-neighborhoods of nodes. It will be instrumental in bounding the discrepancy between vectors $B^{(\ell)}e$ (resp. $B^{(\ell)}\sigma$) and $\{S_{\ell}(i)\}$ (resp. $\{D_{\ell}(i)\}$).  
Its proof, deferred to the Appendix, relies on the previous coupling Lemma~\ref{coupling_lemma_2}. 
\begin{lemma}\label{cycles} 
Assume $\ell=c\log(n)$ with $c\log(\alpha)<1/2$. Then with high probability the number of nodes $i$ whose $\ell$-neighborhood contains one cycle is $O(\log^4(n)\alpha^{2\ell})$. Assume further that $c\log(\alpha)<1/4$. Then with high probability no node $i$ has more than one cycle-edge in its $\ell$-neighborhood.
\end{lemma}

\subsection{From  neighborhood variables $S_t$ and $D_t$ to path matrix $B^{(\ell)}$}
We first state how to transport the deterministic growth controls~(\ref{control_2}) of Theorem~\ref{lemma_LD_bounds} to vectors $B^{(m-1)}e$ and $B^{(m-1)}\sigma$, a key step in the proof of Theorem~\ref{th_ramanujan}. One has the following
\begin{lemma}\label{det_grow2}
Let $\mathcal{B}$ denote the set of nodes $i$ whose $\ell$-neighborhood contains a cycle. Then for $m\le \ell$, $\ell=c\log n$ with $c\log\alpha<1/4$: :
\begin{eqnarray}\label{control_growth1} 
i\notin\mathcal{B}\Rightarrow\left\{
\begin{array}{lll}
(B^{(m-1)} e)_i&=S_{m-1}(i)&=\alpha^{m-1-\ell}(B^{(\ell)} e)_i+O(\log(n))+O(\sqrt{\log(n)\alpha^{m-1}}),\\
(B^{(m-1)}\sigma)_i&=D_{m-1}(i)&=\beta^{m-1-\ell}(B^{(\ell)}\sigma)_i+O(\log(n))+O(\sqrt{\log(n)\alpha^{m-1}}),
\end{array}\right.
\\
\label{control_growth2} 
i\in\mathcal{B}\Rightarrow |(B^{(m)}\sigma)_i|\le |(B^{(m)}e)_i|\le 2 \sum_{t=0}^m S_i(t)=O(\log(n)\alpha^{m}).
\end{eqnarray} 
\end{lemma}
Proof is in the Appendix, together with that of the following Corollary:
\begin{corollary}\label{cor_int}
For all $m\in\{1,\ldots,\ell\}$ it holds with high probability that 
\begin{eqnarray}\label{bound_growth_1} 
\sup_{|x|=1,x'B^{(\ell)}e=x'B^{(\ell)}\sigma=0}|e'B^{(m-1)}x|=O\left(\log^3(n)\alpha^{\ell+m-1}+\sqrt{n}[\log(n)+\sqrt{\log(n) \alpha^{m-1}}]\right),
\\\label{bound_growth_2} 
\sup_{|x|=1,x'B^{(\ell)}e=x'B^{(\ell)}\sigma=0}|\sigma'B^{(m-1)}x|=O\left(\log^3(n)\alpha^{\ell+m-1}+\sqrt{n}[\log(n)+\sqrt{\log(n) \alpha^{m-1}}]\right).
\end{eqnarray}
\end{corollary}
We are now ready to prove Theorem~\ref{th_ramanujan}:
\begin{proof}(of Theorem~\ref{th_ramanujan}).
Using identity (\ref{matrix_expansion}), write for unit norm $x$:
$$
|B^{(\ell)}x|\le \rho(\Delta^{(\ell)})+\sum_{m=1}^{\ell}\rho(\Delta^{(\ell-m)})|\bar{A}B^{(m-1)} x|+\sum_{m=1}^{\ell}\rho(\Gamma^{\ell,m}).
$$
We can ignore the terms $\rho(\Delta^{(\ell)})$ and $\rho(\Gamma^{\ell,m})$, known to be less than $n^{\epsilon}\alpha^{\ell/2}$ from~(\ref{rho_delta}) and~(\ref{rho_gamma}). Recalling the expression~(\ref{bara}) of $\bar{A}$, one has:
$$
|\bar{A}B^{(m-1)}x|\le \frac{a}{n}|B^{(m-1)}x|+O(n^{-1/2}[|\sigma'B^{(m-1)}x|+|e'B^{(m-1)}x| ]).
$$
Using the bounds~(\ref{bound_growth_1},\ref{bound_growth_2}), the right-hand side is no larger than
$$
\frac{a}{n}|B^{(m-1)}x|+O\left(n^{-1/2}\left(\log^3(n)\alpha^{\ell+m-1}+\sqrt{n}[\log(n)+\sqrt{\log(n) \alpha^{m-1}}]\right)\right).
$$
By the previous inequalities~(\ref{control_1},\ref{control_growth1},\ref{control_growth2}) and the row sum bound, we have that 
$$
\rho(B^{(m-1)})=O(\log(n)\alpha^{m-1}).
$$
This thus yields
$$
\begin{array}{ll}
|\bar{A}B^{(m-1)}x|&\le O\left(\frac{\log(n)\alpha^{m-1}}{n}+\frac{\log^3(n)\alpha^{\ell+m-1}}{n^{1/2}}+\log(n)+\sqrt{\log(n) \alpha^{m-1}}\right)
\\
&=O(\log(n)+\sqrt{\log(n) \alpha^{m-1}}).
\end{array}
$$
We thus have
$$
\begin{array}{ll}
|B^{(\ell)}x|&\le n^{\epsilon}\alpha^{\ell/2}+\sum_{m=1}^{\ell}\rho(\Delta^{(\ell-m)})O(\log(n)+\sqrt{\log(n) \alpha^{m-1}})\\
&\le n^{\epsilon}\alpha^{\ell/2}+\sum_{m=1}^{\ell}n^{\epsilon}\alpha^{(\ell-m)/2}O(\log(n)+\sqrt{\log(n) \alpha^{m-1}})\\
&\le n^{\epsilon}\alpha^{\ell/2}O(1+2\ell \log(n)).
\end{array}
$$
The result readily follows.
\end{proof}
We now state two Lemmas which will allow to establish Theorem~\ref{spectral_thm}.
\begin{lemma}\label{lemma_align} 
The following evaluations hold whp for $\ell=c\log n$ with $c\log\alpha<1/4$:
\begin{equation}
\begin{array}{ll}
|B^{(\ell)}e-\{S_{\ell}(i)\}_{i\in\nn}|&=o\left(|B^{(\ell)}e|\right),\\
|B^{(\ell)}\sigma-\{D_{\ell}(i)\}_{i\in\nn}|&=o\left(|B^{(\ell)}\sigma|\right),\\
<B^{(\ell)}e,B^{(\ell)}\sigma>&=o\left(|B^{(\ell)}e|\times |B^{(\ell)}\sigma|\right)
\end{array}
\end{equation}
\end{lemma}
\begin{lemma}\label{lemma_trees}
The following inequalities hold for $\ell=c\log n$ with $c\log\alpha<1/4$:
\begin{eqnarray}\label{norm_bound_1} 
\Omega(\alpha^{\ell}) |B^{(\ell)}e|\le |B^{(\ell)} B^{(\ell)}e|\le O(\log^2 n \alpha^{\ell}) |B^{(\ell)}e|
\\
\label{norm_bound_2} 
\Omega(\beta^{\ell}) |B^{(\ell)}\sigma| \le |B^{(\ell)} B^{(\ell)}\sigma|\le O(\log^{4}(n) \beta^{\ell}) |B^{(\ell)}\sigma|.
\end{eqnarray}
\end{lemma}
Using these, we now establish the following
\begin{theorem}\label{spectral_thm}
For $\ell=c\log n$ with $c\log\alpha<1/4$, the two leading eigenvectors of $B^{(\ell)}$ are asymptotically aligned with vectors $\{S_{\ell}(i)\}$, $\{D_{\ell}(i)\}$, with corresponding eigenvalues of order up to logarithmic terms $\Theta(\alpha^{\ell})$ and $\Omega(\beta^{\ell})$. All other eigenvalues are $O(n^{\epsilon}\sqrt{\alpha^{\ell}})$ for any fixed $\epsilon>0$. 
\end{theorem}
\begin{proof}
Estimates~(\ref{norm_bound_1}--\ref{norm_bound_2}) and the weak Ramanujan property of Theorem~\ref{th_ramanujan} imply that the leading eigenvector is aligned with $B^{(\ell)}e$ and has eigenvalue $\alpha^{\ell}$ up to logarithmic terms. The second eigenvector is necessarily asymptotically in the span of  $\{B^{(\ell)}e,B^{(\ell)}\sigma\}$ and with eigenvalue $\Omega(\beta^{\ell})$. By asymptotic orthonormality of vectors $B^{(\ell)}e$ and $B^{(\ell)}\sigma$ and their asymptotic alignment with $\{S_{\ell}(i)\}$, $\{D_{\ell}(i)\}$ respectively, the conclusion regarding the first two eigen-elements follows. The bound on the magnitude of other eigenvalues follows from Theorem~\ref{th_ramanujan} and the Courant-Fisher theorem. 
\end{proof}
\subsection{Coupling with Poisson tree growth process}\label{sec:coupling}
Introduce the stochastic process $\{V^{\pm}_t\}_{t\ge 0}$ defined by
\begin{equation}\label{tree_process}
\begin{array}{l}
V^+_0=1,\; V^-_0=0,\\
V^+_t,\; V^-_t \hbox{independent conditionally on }\GG_{t-1},\\
\mathcal{L}(V^{\pm}_t|\GG_{t-1})=\hbox{Poi}((a/2)V^{\pm}_{t-1}+(b/2)V^{\mp}_{t-1})
\end{array}
\end{equation}
where $\GG_{t-1}=\sigma(V^{\pm_k},k\le t-1)$. We then have the following
\begin{lemma}\label{coupling_lemma_1}
Let $i\in\nn$ be fixed with spin $\sigma_i=\sigma$. For a constant $c>0$ such that $c\log(\alpha)<1/2$, and $\ell=c\log(n)$, the following holds. The variation distance between $(U^{\pm}_{t}(i))_{t\le \ell}$ and $(V^{\pm\sigma}_{t})_{t\le \ell}$ goes to zero as a negative power of $n$ as $n\to\infty$. 
\end{lemma}
The proof given in the Appendix relies on the Stein-Chen method for Poisson approximation.

Define now the processes 
\begin{equation}
\begin{array}{ll}
M_t&=\alpha^{-t} (V^+_t+V^-_t),\\
\Delta_t&=\beta^{-t}(V^+_t-V^-_t),
\end{array}
\end{equation}
where $V^{\pm}_t$ is as defined in (\ref{tree_process}). We then have the following
\begin{lemma}\label{martingales1}
The two processes $\{M_t\}$, $\{\Delta_t\}$ are $\GG_t$-martingales. Process $\{M_t\}$ is uniformly integrable under Condition $\alpha>1$. Under Condition $\beta^2>\alpha$, process $\{\Delta_t\}$ is also uniformly integrable.
\end{lemma}
\begin{corollary}\label{corollary_martingales} 
Under $\alpha<\beta^2$ the martingale $\{\Delta_t\}$ converges almost surely to a unit mean random variable $\Delta_{\infty}$. Moreover this random variable has a finite variance $1/(\beta^2/\alpha-1)$ to which the variance of $\Delta_t$ converges. It further holds that $\EE|\Delta^2_t-\Delta_{\infty}^2|\to 0$ as $t\to\infty$.
\end{corollary}
Together these properties allow to establish the following
\begin{theorem}\label{thm_local} 
One has the following convergence in probability
\begin{equation}\label{second_moments} 
\lim_{n\to\infty}\frac{1}{n}\sum_{i=1}^n \beta^{-2\ell}D_{\ell}^2(i)=\EE(\Delta_{\infty}^2).
\end{equation}
For all $\tau\in\RR$ that is a point of continuity of the distribution of both $\Delta_{\infty}$ and $-\Delta_{\infty}$, one has the following convergence in probability for both signs $\pm$
\begin{equation}\label{empirical}
\lim_{n\to\infty}\frac{1}{n}\sum_{i\in\nn:\sigma_i=\pm}\II_{y_i\ge \tau/\sqrt{n\EE(\Delta_{\infty}^2)}}=\frac{1}{2}\PP(\pm\Delta_{\infty}\ge \tau).
\end{equation}
Let $y\in\RR^{n}$ be the normed vector defined as
\begin{equation}
y_i=\frac{D_{\ell}(i)}{\sqrt{\sum_{j=1}^n D_{\ell}(j)^2}},\; i=1,\ldots,n.
\end{equation}
Let $x$ be a vector in $\RR^n$ such that we have the convergence in probability 
\begin{equation}
\lim_{n\to\infty}||x-y||=0.
\end{equation}
Define the spin estimates 
\begin{equation}
\hat{\sigma}_i=\left\{\begin{array}{ll}
+&\hbox{if }x_i\ge t/\sqrt{n\EE(\Delta_{\infty}^2)},\\
-&\hbox{ otherwise.}
\end{array}\right.
\end{equation}
For each $t$ that is an atom of neither $\Delta_{\infty}$'s or $-\Delta_{\infty}$'s distribution, the following convergence in probability holds
\begin{equation}\label{empirical_overlap}
\lim_{n\to\infty}\frac{1}{n}\sum_{i=1}^n\sigma_i \hat{\sigma}_i=\frac{1}{2}\left(\PP(\Delta_{\infty}\ge t)-\PP(-\Delta_{\infty}\ge t)\right).
\end{equation}
\end{theorem}
To convey the main ideas of the proof (deferred to the Appendix), we now indicate how to establish a property similar to~(\ref{empirical}), namely for a continuous bounded function $g$ we establish convergence in probability 
\begin{equation}\label{test} 
\frac{1}{n}\sum_{i\in\nn"'\sigma_i=\pm}g(\beta^{-\ell}D_{\ell}(i))\to \frac{1}{2} \EE g(\pm \Delta_{\infty}).
\end{equation}
The expectation of the sum in the left-hand side reads
$$
\frac{\npm}{n}\left[\EE g(\pm\Delta_{\ell})\II_{\hbox{coupling}}+O(|g|_{\infty})(1-\PP(\hbox{coupling}))\right]
$$
where the event $\{\hbox{coupling}\}$ refers to coupling as in Lemma~\ref{coupling_lemma_1}. By this Lemma, the martingale convergence property of $\Delta_{\ell}$, the fact that $\npm/n\to 1/2$ and continuity of $g$, this expectation converges to $(1/2)\EE g(\pm \Delta_{\infty})$. Now the expectation of the square of the empirical average in~(\ref{test}) reads
$$
\frac{1}{n^2}\left[\npm O(|g|^2_{\infty})+2\binom{\npm}{2} \EE g(\beta^{-\ell}D_{\ell}(i))g(\beta^{-\ell}D_{\ell}(j))\right]
$$
where $i\ne j$ are two fixed nodes with spin $\pm$. By the coupling lemma~\ref{coupling_lemma_2} it holds that 
$$
\EE g(\beta^{-\ell}D_{\ell}(i))g(\beta^{-\ell}D_{\ell}(j))=\left(\EE g(\beta^{-\ell}D_{\ell}(i))\right)^2+O(n^{-\epsilon})|g|^2_{\infty}.
$$
It follows that the variance of the empirical average in~(\ref{test}) goes to zero as $n\to\infty$. Its announced convergence in probability to $(1/2)\EE g(\pm \Delta_{\infty})$ then follows by Tchebitchev's inequality.

Theorems~\ref{spectral_thm} and ~\ref{thm_local} readily imply Theorem~\ref{th_1}.

\section{Conclusions}\label{sec:conclusions}
The methods developed here may find further applications, e.g. to prove the more general conjecture by Heimlicher et al.~\cite{NIPS12} of a phase transition in the labeled stochastic block model. More generally one might ask what is the range of applicability of our path expansion approach to ``fix'' spectral methods by recovering Ramanujan-like spectral separation properties. It is likely that a similar regularization would occur by considering matrix $\hat{B}$ defined by $\hat{B}_{ij}=\II_{d_{\g}(,ij)=\ell}$ but we have not been able to prove this yet.

{\bf acknowledgements:} The author gratefully acknowledges stimulating discussions on the topic with Marc Lelarge and Charles Bordenave.
\bibliography{BibCommunityDetection}
\newpage
\begin{appendix}
\section{Proof of Proposition~\ref{trace_delta}}
\begin{proof}We control the spectral radius of $\Delta^{(\ell)}$ using the trace method. Specifically chose $k>0$. It holds that
\begin{equation}
\EE(\rho(\Delta^{(\ell)})^{2k})\le \EE \hbox{Tr}(\Delta^{(\ell)})^{2k}).
\end{equation}
Note that $\hbox{Tr}(\Delta^{(\ell)})^{2k})$ is the sum over circuits of length $2k$ of the products of the entries $\Delta^{(\ell)}_e$ over the edges $e$ in the circuit. Moreover, given the definition of $\Delta^{(\ell)}$, these correspond to products of entries $A_{e}-\bar{A}_e$ over edges $e$ of circuits of length $2k\ell$ satisfying the property that consecutive length $\ell$-paths are simple.

We bound the expectation of the corresponding sum as follows. Let $v$ (respectively, $e$) be the number of nodes (respectively, edges) traversed by a particular circuit. The quantity $c=e-v+1$ is the so-called ``tree excess'', counting the number of edges that are traversed while not being part of the tree consisting of edges whose first traversal strictly augments the number of spanned nodes.

We represent the corresponding circuit as follows. 

We number nodes by the order in which they are met by the circuit, starting with node 1.

We break each length $\ell$-simple path into consecutive sequences consisting of 
\begin{itemize} 
\item a path using only edges already used in the circuit, and lying on the tree of new node discoveries
\item a path of discoveries of new nodes
\item a cycle edge connecting the end of the two previous steps to a node already spanned. Such a cycle edge may have already been traversed by the circuit one or several times.
\end{itemize}
Given the tree previously spanned, and the current position on it, the first part of the sequence is characterized by the node label of its end: indeed, since on this subsequence we enforce the condition that the paths are simple, back-tracking is forbidden. Hence there is only one path on the tree going from the origin to the destination. We thus represent the first part by the number of the destination node if this part is non-empty, and by zero otherwise. 

The second part of the sequence is simply represented by its length, which is constrained to lie in $\{0,\ldots,\ell\}$. Indeed, it cannot exceed $\ell$, as we consider sequences that lie within a length $\ell$-simple path. 

Finally, the third part of the sequence is simply characterized by the number characterizing its end point, and by zero if this part is not present. We must allow for this case, as when we break up a length $\ell$-simple path into constituting such sequences, the last such sequence may not end up by traversal of such a redundant edge.

Let us now use this representation to bound the number of corresponding sequences. An individual sequence is represented by a triplet $(p,q,r)$ with $p\in\{0,\ldots,v\}$, $q\in\{0,\ldots,\ell\}$, and $r\in\{0,\ldots,v\}$. Note further that each such sequence corresponds to either the end of an individual length $\ell$-simple path, or the traversal of a redundant edge. The number of such edges is $c=e-v+1$, and each edge can be traversed at most $2k$ times by the constraint that circuits are formed from length $\ell$ simple paths. Thus the number of valid circuits corresponding to $v$ and $e$ is at most
$$
[(v+1)^2(\ell+1)]^{2k(1+e-v+1)}.
$$ 
For a given number of nodes $v$ and edges $e$, the number of corresponding nodes in $\{1,\ldots,n\}$ is upper-bounded by $n^v$. For a given edge present with multiplicity $m\in\{1,\ldots, 2k\}$, the corresponding expectation is zero if $m=1$, and for $m\ge 2$, we have
$$
\EE((\xi_{ij} -\EE(\xi_{ij}|\sigma))^m|\sigma)\le \frac{a(\sigma_i,\sigma_j)}{n}\cdot
$$ 
Here $a(\sigma_i,\sigma_j)$ equals $a$ if $\sigma_i=\sigma_j$ and $b$ otherwise.
For the cyclic edges we use the upper bound $\max(a,b)/n$. There are $e-v+1$ such factors. We are left with a tree with $v-1$ edges, for which upon averaging over $\sigma$ we get a contribution $(\alpha/n)^{v-1}$. Now the number of nodes $v$ on any configuration whose contribution in expectation does not vanish must lie between $\ell+1$ and $k\ell+1$: indeed each node discovery costs one edge, but this edge must be doubled for the contribution not to vanish. Since there are in total $2k\ell$ edges, at most $k\ell$ nodes can be discovered in addition to the original node of the circuit. The number of distinct edges is similarly bounded by $k\ell$ in any configuration with non-vanishing expectation.  This gives the bound
$$
\EE(\rho(\Delta^{(\ell)})^{2k})\le\sum_{v=\ell+1}^{k\ell+1}\sum_{e=v-1}^{k\ell}[(v+1)^2(\ell+1)]^{2k(1+e-v+1)} n^{v}\left(\frac{\alpha}{n}\right)^{v-1}\left(\frac{a\vee b}{n}\right)^{e-v+1}
$$
as announced.
\end{proof}
\section{Proof of Proposition~\ref{trace_gamma}}
\begin{proof}
By the trace method again, we have that $ \rho(\Gamma^{\ell,m})^{2k}$ is upper-bounded by the trace of $[\Gamma^{\ell,m}(\Gamma^{\ell,m})']^{k}$. The latter corresponds to sum over circuits of length $2\ell k$ of products of terms that can be either $A_e-\bar{A}_e$, $\bar{A}_e$ or $A_e$. The constraints are that a length $\ell$ chunk of the circuit is the concatenation of two simple paths of length $m-1$ and $\ell-m$, and that the two of them have a non-empty intersection.

We represent such contributions as follows. We let $v$ denote the number of nodes traversed by the circuit, and by $e$ the number of edges traversed by the circuit, while ignoring edges that are weighed by an $\bar{A}$-term. Note that by the constraint that the concatenated simple parts of each length $\ell$-chunk intersect, the corresponding graph is necessarily connected. We adopt the following representation of the corresponding circuits.

Nodes are again denoted by the order in which they are first met, starting with node 1. We represent each simple path that constitutes the circuit by sequences of three phases as before. Note that there are now $4k$ such simple paths: each length $\ell$ chunk of the original circuit is broken into an $m-1$- and an $\ell-m$-path. We adopt the same representation as before, except that we must now also incorporate the label of the starting point after traversal of an $\bar{A}$-edge. 

Thus we have the upper bound on the number of valid circuit labels with $v$ nodes and $e$ edges:
$$
v^{2k}[(v+1)^2(\ell+1)]^{4k(1+e-v+1)}.
$$
Let us bound the values that $v$ and $e$ can take. Necessarily, $v\ge \max(m,\ell-m+1)$: indeed, each length $\ell$ chunk comprises simple paths of length $m-1$ and $\ell-m$. Moreover, there are overall $2k(\ell-1)$ edges (recall that we discount the $\bar{A}$-edges). Out of these, $2k (\ell-m)$ must be doubled for the expectation not to vanish. There are thus at most $1+k(\ell+m)$ nodes $v$ in total, and at most $k(\ell+m)$ distinct edges in total. 

We thus obtain the upper bound
$$
\EE[(\rho(\Gamma^{\ell,m})^{2k}]\le\sum_{v=m\vee (\ell-m+1)}^{1+k(\ell+m)}\sum_{e=v-1}^{k(\ell+m)}\left(\frac{a\vee b}{n}\right)^{2k} v^{2k}[(v+1)^2(\ell+1)]^{4k(1+e-v+1)}
n^v \left(\frac{\alpha}{n}\right)^{v-1} \left(\frac{a\vee b}{n}\right)^{e-v+1}.
$$
\end{proof}
\section{Proof of Theorem~\ref{lemma_LD_bounds}}
The following inequality is easily verified to hold for any non-negative $U$, $V$, $a$, $b$, $n$ such that $a/n,b/n\le 1$, and will be instrumental in the sequel:
\begin{equation}\label{basic_ineq} 
\frac{a U+b V}{n}-\frac{1}{2}\left(\frac{a U+b V}{n}\right)^2\le 1-(1-a/n)^U(1-b/n)^V\le \frac{a U+b V}{n}.
\end{equation}
Next lemma is the key ingredient to establish Theorem~\ref{lemma_LD_bounds}. 
\begin{lemma}
The following properties hold with high probability for all $i\in \nn$ and all $t\le\ell$, with $\ell=C\log(n)$.

Let $T=\inf\{t\le\ell: S_t\ge K\log(n)\}$ for some constant $K$. Then $S_T=\Theta(\log(n))$.

Let $\epsilon_t:=\epsilon \alpha^{-(t-T)/2}$ for some constant $\epsilon>0$. Then for all $t,t'\in\{T,\ldots,\ell\}$, $t>t'$, the vector $U_t=(U^+_t,U^-_t)$ verifies the coordinate-wise bounds:
\begin{equation}\label{ineq_prod} 
U_t\in\left[\prod_{s=t'+1}^t(1-\epsilon_s) M^{t-t'+1}U_{t'}, \prod_{s=t'+1}^t(1+\epsilon_s) M^{t-t'+1} U_{t'}\right],
\end{equation}
where $M$ denotes the matrix $(a/2\;  b/2, b/2\; a/2)$.
\end{lemma}
\begin{proof}
Recall that conditionally on $\mathcal{F}_{t-1}$ the random variables $U^+_t$ and $U^-_t$ are independent, distributed according to 
$$
U^{\pm}_t\sim\hbox{Bin}\left(n^{\pm}-U^{\pm}_{<t},1-(1-a/n)^{U^{\pm}_{t-1}}(1-b/n)^{U^{\mp}_{t-1}}\right).
$$
Let $T$ be the first instant $t$ for which $U_t\ge K\log(n)$, for some $K$ to be specified. 

By definition of $T$, necessarily $U_{T-1}<K\log(n)$. Thus
$$
U^{\pm}_{T}\le\hbox{Bin}(n^{\pm},(a\vee b)\frac{K\log(n)}{n}).
$$
The mean of the Binomial distribution in the right-hand side of the above is equivalent to $(a\vee b)(1/2)K\log(n)$ and less than $\kappa \log(n)$ for $\kappa=(a\vee b)K$. 
Hence by Chernoff's inequality, for $h(x):=x\log(x)-x+1$,
$$
\PP(U^{\pm}_T\le K'/2 \log(n)|\mathcal{F}_{T-1})\le e^{-\kappa \log(n)h(K'/2\kappa)}.
$$
Take $K'$ so that $\kappa h(K'/2\kappa)>2$. The right-hand side of the above is then no larger than $n^{-2}$.

Thus properties (\ref{control_1}) clearly hold for $t\le T$. We now establish that they hold with sufficiently large probability for larger $t$. 

Conditional on $\mathcal{F}_T$, the binomial distribution of $U^{\pm}_{T+1} $ has mean
$$
[n^{\pm}-U^{\pm}_{<T+1}]\times[1- (1-a/n)^{U^{\pm}_T}(1-b/n)^{U^{\mp}_T}].
$$
Using the inequalities (\ref{basic_ineq}) we obtain that this mean lies in the interval
$$
[(a\wedge b)\frac{1}{3}K\log(n),(a\vee b)K'\log(n)].
$$
For a given $\epsilon>0$, we can choose $K$ sufficiently large so that 
$$
(a\wedge b)\frac{1}{3}K h(1+\epsilon)>2.
$$
It follows that $U^{\pm}_{T+1}$ admits a relative deviation from its conditional mean by $\epsilon$ with probability at most $n^{-2}$. 

We now define
$$
\epsilon_t=\epsilon \alpha^{-(t-T)/2},
$$
and consider the events $\mathcal{A}_t:=\{U^{\pm}_t\in[1-\epsilon_t,1+\epsilon_t]\frac{a U^{\pm}_{t-1}+b U^{\mp}_{t-1}}{2}\}$.
Conditionally on $\mathcal{A}_T,\ldots,\mathcal{A}_{t}$, the vector $U_t=(U^+_t,U^-_t)$ verifies the announced inequality (\ref{ineq_prod}). 
 Given that $\alpha$ is the spectral radius of $M$, it follows from this condition that $U^{\pm}_t\ge (1-O(\epsilon))\alpha^{t-T}K'' \log(n)$. We then check that Chernoff's bound applies to show that the condition holds at step $t$ with high enough probability. It suffices to ensure that 
$$
U^{\pm}_t \tilde{h}(\epsilon_t)\ge 2\log(n),
$$
where we take $\tilde{h}(u):=\min[(1+u)\log(1+u)-u,(1-u)\log(1-u)+u)$. 
However as we just saw the left-hand side of this expression is lower-bounded by
$$
(1-O(\epsilon))\alpha^{t-T} K''\log(n) \tilde{h}(\epsilon_t)\ge (1-O(\epsilon))\alpha^{t-T} K''\log(n) K''' \epsilon_t^2,
$$
where we took a second-order expansion of $\tilde{h}$ around 0.
The condition is therefore met as soon as $(1-O(\epsilon))K'' K''' \epsilon^2 \ge 2.$ For $K$ large enough this holds.
\end{proof}
\begin{proof}(of Theorem~\ref{lemma_LD_bounds}). 
For $t\le \ell$, if $t\le T$, we necessarily have that $S_t, |D_t|=O(\log n)$. Consider then $t>T$. Note that matrix $M$ is such that 
$$
M^{k}=\frac{1}{2}\left(\begin{array}{cc} \alpha^k+\beta^k&\alpha^k-\beta^k\\ \alpha^k-\beta^k&\alpha^k+\beta^k
\end{array}
\right)\cdot
$$
Using (\ref{ineq_prod}), we readily have for $t,t'\le T$, with $t>t'$:
$$
\begin{array}{ll}
S_t&\le \prod_{s=t'+1}^t(1+\epsilon_s)(1,1) M^{t-t'} U_{t'}\\
&=\prod_{s=t'+1}^t (1+\epsilon_s) \alpha^{t-t'}S_{t'}.
\end{array} 
$$
A similar lower bound holds with $-\epsilon_s$ in place of $+\epsilon_s$. Setting $t'=T$ in the upper bound, since $S_T=O(\log(n)$, the upper bound (\ref{control_1}) follows for $S_t$, as $\prod_{s=T+1}^t(1+\epsilon_s)=O(1)$.

Note now that 
$$
\max\left(\prod_{s=t'+1}^t (1+\epsilon_s)-1,1-\prod_{s=t'+1}^t (1-\epsilon_s)\right)=O(\epsilon_{t'})=O(\alpha ^{-t'/2}).
$$
It readily follows that (\ref{control_2}) holds for $S_t$.

Consider now $D_t$. Using (\ref{ineq_prod}) again, we have

$$
D_{t-1}-\alpha \epsilon_t S_t \le D_t \le \beta D_{t-1}+\alpha \epsilon_t S_t.
$$
Iterating, we obtain
\begin{equation}\label{bound_D}
|D_t-\beta^{t-t'}D_{t'}|\le \sum_{s=t'+1}^t \alpha \beta^{t-s} \epsilon_s S_s.
\end{equation}
Since $S_s=O(\log(n) \alpha^{s-T})$, $|D_T|=O(\log(N)$ and $\epsilon_s=O(\alpha^{-(s-T)/2})$, we obtain for $t'=T$:
$$
|D_t|=O(\log(n)\beta^t+\sum_{s=T+1}^t \beta^{t-s}\log(n) \alpha^{(s-T)/2})=O(\log(n)\beta^t),
$$
where we have used the assumption that $\beta^2>\alpha$ to bound $\sum_{u>0}\beta^{-u}\alpha^{u/2}$. Property (\ref{control_1}) thus holds for $D_t$.

Finally, the right-hand side of (\ref{bound_D}) is of order
$$
\sum_{s=t'+1}^t  \beta^{t-s}\alpha^{(s-T)/2} \log(n)=O(\log(n)\beta^{t-t'}\alpha^{t'/2}).
$$
Thus setting $t=\ell$, for $\ell>t'\ge T$ we have
$$
D_{t'}=\beta^{t'-\ell}D_{\ell}+O(\log(n)\alpha^{t'/2}).
$$
Since for $t'<T$ we readily have $D_{t'}=O(\log(n)$ by definition of $T$, property (\ref{control_2}) follows for $D_t$.
\end{proof}
\section{Proof of Lemma~\ref{cycles}}
\begin{proof}
There are two ways for creating cycles within the distance $k$-neighborhood of $i$: an edge may be present between two nodes at distance $k-1$ of $i$, or two nodes at distance $k-1$ may be connected to the same node at distance $k$ of $i$. The number of edges of the first type is stochastically dominated by $\hbox{Bin}(S_{k-1}^2, a\vee b/n)$. Its expected number conditionally on $\Omega_{k-1}(i)$, defined as
$$
\Omega_{k-1}(i)=\{S_{k-1}(i)\le C(\log n) \alpha^{k-1}\}
$$ is at most $O(\log^2(n) \alpha^{2\ell}/n)$. Thus by the union bound the probability that there is such an edge in the $\ell$-neighborhood of $i$ is at most:
$$
\ell \times O(\log^2(n) \alpha^{2\ell}/n)+\sum_{k=1}^{\ell}(1-\PP(\Omega_k(i)))=O(\log^3(n)\alpha^{2\ell}/n).
$$
As for the second type of cycles, its number is stochastically dominated by 
$$
\hbox{Bin}(n,(a\vee b/n)^2 S_{k-1}^2).
$$
On $\Omega_{k-1}(i)$ its conditional expectation is $O( \log^2(n)\alpha^{2\ell})$. 

By the same argument, the probability that there are two cycle-edges within the $\ell$-neighborhood of $i$ is upper-bounded by $O(\log^6(n)\alpha^{4\ell}/n^2)$. By the union bound we readily have that with high probability no node has two cycle-edges within its $\ell$-neighborhood as soon as $\log^6(n)\alpha^{4\ell}\ll n$, which holds for $\ell=c\log(n)$ with $c\log(\alpha)<1/4$.

Let $Z_i$ denote the event that the $\ell$-neighborhood of $i$ contains a cycle. On the event 
$$
\cap_{i\in \nn}\cap_{k\le \ell}\Omega_k(i),
$$
 the $\ell$-neighborhoods of an arbitrary pair of distinct nodes $i$, $j$ are disjoint with probability $1-O(\log^2(n)\alpha^{2\ell}/n)$, conditionally upon which the probability that they both have a cycle in their neighborhood is upper-bounded by $O(\log^6(n)\alpha^{4\ell}/n^2)$. Conditionally on the event that their neighborhoods meet, the expectation of the product $Z_i Z_j$ is still upper-bounded by $O(\log^3(n)\alpha^{2\ell}/n)$.

Eventually Markov's inequality yields
$$
\begin{array}{ll}
\PP(\sum_i Z_i\ge m \log^3(n)\alpha^{2\ell})&\le \frac{\EE(\sum_i Z_i)^2}{m^2 \log^6(n)\alpha^{4\ell}}\\
&\le 
\frac{ n\EE(Z_1)+n^2\EE(Z_1 Z_2)}{m^2 \log^6(n)\alpha^{4\ell}}\\
&\le
\frac{O(\log^3(n)\alpha^{2\ell})+n^2[O(\log^6(n)\alpha^{4\ell}/n^2)+(\log^2(n)\alpha^{2\ell}/n)(\log^3(n)\alpha^{2\ell}/n)}{m^2 \log^6(n)\alpha^{4\ell}}\\
&=O(\frac{1}{m^2}).
\end{array}
$$
Taking $m=\log(n)$ (say), then with high probability $\sum_i Z_i=O(\log^4(n) \alpha^{2\ell})$.
\end{proof}
\section{Proof of Lemma~\ref{det_grow2}}
Consider first those indices $i\notin\mathcal{B}$ whose $\ell$-neighborhood is a tree. For any $k$ and any $m\le \ell$, $B^{(m)}_{ik}$ can take values only 0 or 1 (there is at most one simple path connecting $i$ to $k$), and $B^{(m)}_{ik}=1$ if and only if $d_{\g}(i,k)=m$. For such $i$, one therefore has the following identities:
\begin{equation}
i\notin\mathcal{B}\Rightarrow\left\{
\begin{array}{ll}
(B^{(m)} e)_i=S_i(m),\\
(B^{(m)}\sigma)_i=D_i(m),
\end{array}\right.
\end{equation}
Relations~(\ref{control_growth1}) readily follows  from Theorem~\ref{lemma_LD_bounds}.
Let us now consider an index $i\in\mathcal{B}$ whose $\ell$-neighborhood is not a tree. We know from Lemma~\ref{cycles} that such nodes have in their $\ell$-neighborhood only one cycle. Clearly only nodes at distance at most $\ell$ of $i$ can be counted in $(B^{(\ell)}e)_i$, and they can be counted at most twice because the neighborhood contains only one cycle.
Control~(\ref{control_growth2}) readily follows. 
\section{Proof of Corollary~\ref{cor_int}}
\begin{proof}
Let ${\mathcal B}$ denote the set of nodes $i$ such that their $\ell$-neighborhood contains a cycle. Let $x$ be a normed vector such that $x'B^{(\ell)}e=0$. We then have
\begin{equation}\label{display42} 
\begin{array}{lll}
|e'B^{(m-1)}x|&=&|\sum_{i\in\nn}x_i (B^{(m-1)}e)_i|\\
&\le&|\sum_{i\in{\mathcal{B}}} x_i(B^{(m-1)}e)_i|\\
&&+|\sum_{i\in{\mathcal{B}}} x_i \alpha^{m-1-\ell} (B^{(\ell)}e)_i|\\
&&+|\sum_{i\in\nn} x_i[\alpha^{m-1-\ell} (B^{(\ell)}e)_i+O(\log(n)+\sqrt{\log(n)\alpha^{m-1}})|.
\end{array}
\end{equation}
Using the bound (\ref{control_growth2}) for $i\in {\mathcal{B}}$, we can bound the first summation, using Cauchy-Schwarz's inequality by
$$
\begin{array}{ll}
|\sum_{i\in{\mathcal{B}}} x_i[(B^{(m-1)}e)_i|&\le O(\log(n)\alpha^{m-1})\sqrt{|{\mathcal{B}}|}\\
&\le O(\log^3(n) \alpha^{\ell+m-1}),
\end{array}
$$
where we have used the bound on the size of ${\mathcal{B}}$ derived in Lemma~\ref{cycles}. The second summation in(\ref{display42}) is similarly bounded. As for the third summation, using the fact that $e'B^{(\ell)}x=0$, it is upper-bounded by
$$
|\sum_{i\in \nn}x_i O(\log(n)+\sqrt{\log(n)\alpha^{m-1}})|.
$$
By Cauchy-Schwarz again, this is no larger than
$$
O\left(\sqrt{n}(\log(n)+\sqrt{\log(n)\alpha^{m-1}})\right). 
$$
The announced bound (\ref{bound_growth_1}) on $|e'B^{(m-1)}x|$ follows. Similarly, the bound (\ref{bound_growth_2}) on $|\sigma' B^{(m-1)}x|$ follows by using property $\sigma'B^{(\ell)}x=0$ instead of property $e' B^{(\ell)}x=0$.
\end{proof}
\section{Proof of Lemma~\ref{coupling_lemma_1}}
\begin{proof}
We assume that $\sigma_i=+$, the case $\sigma=-$ being similar. Introduce the events 
\begin{equation}\label{omega} 
\Omega_k=\{S_k\le C\log(n)\alpha^k\} ,\; k\ge 1,
\end{equation}
where constant $C$ is as in Theorem~\ref{lemma_LD_bounds}. As established in the proof of Theorem~\ref{lemma_LD_bounds}, the probability of each $\Omega_k$ is $1-o(n^{-2})$. 

Let us evaluate, conditionally on ${\mathcal F}_{k-1}$ and on $\Omega_{k-1}$ the variation distance between $(U^+_k,U^-_k)$ and a pair of (conditionally on ${\mathcal F}_{k-1}$) independent random variables with respective distributions 
$$
W^+_k\sim\hbox{Poi}\left(\frac{a U^+_{k-1}+b U^-_{k-1}}{2}\right), \; W^-_k\sim \hbox{Poi}\left(\frac{a U^-_{k-1}+b U^+_{k-1}}{2}\right).
$$
The Stein-Chen method enables to bound the variation distance between a $\hbox{Bin}(n,\lambda/n)$ and a $\hbox{Poi}(\lambda)$ random variables by $n\min(1,\lambda^{-1})(\lambda/n)^2\le \lambda/n$. Furthermore, two Poisson random variables with respective parameters $\lambda$, $\lambda'$ have variation distance at most $|\lambda-\lambda'|$. This entails the bounds
$$
\begin{array}{lll}
d_{var}(U^+_k,W^+_k)&\le & \left[1-(1-a/n)^{U^+_{k-1}}(1-b/n)^{U^-_{k-1}}\right]\\
&&+\left| (\np-U^+_{\le k-1})[1-(1-a/n)^{U^+_{k-1}}(1-b/n)^{U^-_{k-1}}] - \frac{a U^+_{k-1}+b U^-_{k-1}}{2}\right|.
\end{array}
$$
We now use (\ref{basic_ineq}) to obtain
$$
\begin{array}{lll}
d_{var}(U^+_k,W^+_k)& \le & \frac{a U^+_{k-1}+b U^-_{k-1}}{n}\\
&&+ |\np-U^+_{\le k-1}-n/2|[1-(1-a/n)^{U^+_{k-1}}(1-b/n)^{U^-_{k-1}}]\\
&&+ \frac{1}{4n}(a U^+_{k-1}+bU^-_{k-1})^2.
\end{array}
$$
Let us now specify constant $c$ such that $c\log(\alpha)<1/2$, i.e. there is $\epsilon>0$ such that $c\log(\alpha)\le 1/2-\epsilon$. 
For $k\le\ell=c\log(n)$, on $\Omega_{k-1}$, it holds that $S_{k-1}\le C \log(n) n^{1/2-\epsilon}$. This, together with the bound $|n/2-\np|\le O(\log(n)n^{1/2})$ ensures the upper bound
$$
d_{var}(U^+_k,W^+_k)\le O(\log(n)n^{-\epsilon}).
$$
The same bound is readily established for the variation distance $d_{var}(U^-_k,W^-_k)$. These bounds readily establish by induction that the variation distance between the two sequences $(U^{\pm}_k)_{k\le\ell}$ and $(V^{\pm}_k)_{k\le\ell}$ are upper-bounded by $O(\log^2(n)n^{-\epsilon})$, establishing the Lemma.
\end{proof}
\section{Proof of Lemma~\ref{martingales1}}
\begin{proof}
Write 
$$
\begin{array}{ll}
\EE(V^{\pm}_{t+1}|\GG_t)&=\frac{a}{2}V^{\pm}_{t}+\frac{b}{2}V^{\mp}_{t}\\
&=\frac{a}{2}\frac{\alpha^{t}M_{t}\pm\beta^t\Delta_{t}}{2}+\frac{b}{2}\frac{\alpha^t M_t\mp\beta^t \Delta_t}{2}\\
&=\alpha^{t+1}M_t\pm\beta^{t+1}\Delta_t.
\end{array}
$$
It readily follows that both processes $\{M_t\}$, $\{\Delta_t\}$ are martingales. To establish uniform integrability we shall show that both processes have uniformly bounded variance. To that end we use the conditional variance formula
$$
\hbox{Var}(X)=\hbox{Var}(\EE(X|\FF))+\EE(\hbox{Var}(X|\FF)),
$$
and the fact that the variance of a Poisson random variable equals its mean. Thus
$$
\EE(V^{\pm}_t|\GG_{t-1})=\hbox{Var}(V^{\pm}_t|\GG_{t-1})=\frac{a V^{\pm}_{t-1}+bV^{\mp}_{t-1}}{2}.
$$
This yields by the conditional variance formula
$$
\begin{array}{ll}
\hbox{Var}(M_t)&=\hbox{Var}(M_{t-1})+\EE(\alpha^{-t}M_{t-1})\\
&=\hbox{Var}(M_{t-1})+\alpha^{-t}.
\end{array}
$$
Since $\hbox{Var}(M_0)=0$, it follows by induction that
$$
\hbox{Var}(M_t)=\frac{1-\alpha^{-t}}{\alpha-1},\; t>0.
$$ 
The latter is uniformly bounded for $\alpha>1$ hence the uniform integrability of $\{M_t\}$ under this condition.

Write now
$$
\begin{array}{ll}
\hbox{Var}(\Delta_t)&=\hbox{Var}(\Delta_{t-1})+\EE(\beta^{-2t}\hbox{Var}(V^+_t-V^-_t|\GG_{t-1}))\\
&=\hbox{Var}(\Delta_{t-1})+\EE(\beta^{-2t}\alpha^t M_{t-1})\\
&=\hbox{Var}(\Delta_{t-1})+\beta^{-2t}\alpha^t.
\end{array}
$$
It thus follows by $\hbox{Var}(\Delta_0)=0$ and induction that
$$
\hbox{Var}(\Delta_t)=\frac{1-(\alpha/\beta^2)^t}{\beta^2/\alpha-1},\; t>0,
$$
thus establishing uniform integrability of martingale $\{\Delta_t\}$. 
\end{proof}
\section{Proof of Corollary~\ref{corollary_martingales}}
\begin{proof}
Convergence almost surely and in $L_1$ is guaranteed under uniform integrability by the martingale convergence theorem (\cite{williams91}). Finiteness of the limiting variable's variance under uniform bounds on the variance is also standard; it follows from Fatou's lemma. Convergence of the variances is established as follows. The limiting variable satisfies a distributional equation given by
\begin{equation}
\Delta=\beta^{-1}\left(\sum_1^{\hbox{Poi}(a/2)}\Delta_i-\sum_{i=1}^{\hbox{Poi}(b/2)}\Delta'_i\right)
\end{equation}
where the $\Delta_i$, $\Delta'_i$ are i.i.d. and distributed as $\Delta$. The only solution for the variance of $\Delta$, apart from the degenerate solution $0$, is then readily seen to be $1/(\beta^2/\alpha-1)$, which is indeed the limit of the variance of $\Delta_t$. The $L_1$-convergence of $\Delta^2_t$ to $\Delta_{\infty}^2$ is then a direct consequence of Scheff\'e's lemma.
\end{proof}
\section{Proof of Theorem~\ref{thm_local}}
\begin{proof}
Note that with probability of order $1-O(n^{-\epsilon})$ for fixed positive $\epsilon$, $\sigma(i)\beta^{-\ell} D_{\ell}(i)$ coincides with $\Delta_{\ell}$ by the coupling lemma~\ref{coupling_lemma_1}. When this coupling fails, by the bounds established in Theorem~\ref{lemma_LD_bounds}, it holds that $\beta^{-\ell}D_{\ell}(i)$ is $O(\log(n))$. This entails that the left-hand side of (\ref{second_moments}) verifies
$$
\EE\left(\frac{1}{n}\sum_{i=1}^n \beta^{-2\ell}D_{\ell}^2(i)\right)=O(\log^2(n))n^{-\epsilon}+\EE(\Delta_{\ell}^2 \II_{coupling}).
$$
Write
$$
\begin{array}{ll}
|\EE(\Delta_{\ell}^2 \II_{coupling})-\EE(\Delta_{\infty}^2|&\le \EE|\Delta_{\ell}^2-\Delta_{\infty}^2|+\EE(\Delta_{\infty}^2\II_{no\;coupling}).
\end{array}
$$
By Corollary~\ref{corollary_martingales}, we have that the first term in the right-hand side goes to with $\ell$; the second term goes to zero with the probability that coupling fails (e.g. using Hardy-Littlewood-Polya's rearrangement inequalities). Thus the expectation converges to $\EE \Delta_{\infty}^2$.

Let us now consider the second moment of the empirical sum:
$$
\EE\left(\frac{1}{n}\sum_{i=1}^n\beta^{-2\ell}D_{\ell}^2\right)^2.
$$
We break it into two terms, the first being
$$
\frac{1}{n^2}\EE\sum_{i=1}^n \beta^{-4\ell} D_{\ell}^4(i).
$$
Using Lemma~\ref{coupling_lemma_1} and Theorem~\ref{lemma_LD_bounds}, using similar arguments as before we can bound this term by
$$
\frac{1}{n}O(\log^2(n))\left[\EE(\Delta_{\infty}^2)+o(1)\right]
$$
which clearly goes to zero as $n\to\infty$.

The second term is given by
$$
\frac{2}{n^2}\sum_{i<j}\beta^{-4\ell}\EE(D_{\ell}^2(i)D_{\ell}^2(j)).
$$
For given $i<j$, using Lemma~\ref{coupling_lemma_2} and Theorem~\ref{lemma_LD_bounds}, we bound the product $D_{\ell}^2(i)D_{\ell}^2(j)$ by $O(\log^4(n)\beta^{4\ell}$ on the event that the coupling with independent copies fails, and by $\EE(D_{\ell}^2(i))\EE(D_{\ell}^2(j))$ on the event that it succeeds. We then bound each of the individual terms in this product as in the control of the expectation done in the first half of the proof, thus obtaining the upper bound for the second moment:
$$
o(1)+\frac{2}{n^2}\sum_{i<j} \left[O(n^{-2\epsilon} \log^4(n))+(\EE(\Delta_{\infty}^2))^2+o(1)\right].
$$
It readily follows that 
$$
\lim_{n\to\infty}\EE\left[\frac{1}{n}\sum_{i=1}^n\beta^{-2\ell}D_{\ell}^2-\EE(\Delta_{\infty}^2)\right]^2=0.
$$

The convergence in probability~(\ref{second_moments}) follows. 

We now turn to establishing~(\ref{empirical}). We shall only consider the case of sign +, the other being handled similarly. Fix some arbitrarily small $\delta>0$. Because $\tau$ is a continuity point of the distribution of $\Delta_{\infty}$, we can find two bounded Lipschitz-continuous functions $f$, $g$ such that 
$$
f(u)\le \II_{u\ge \tau}\le g(u), \; u\in \RR
$$
and 
$$
0\le \EE(g(\Delta_{\infty})-f(\Delta_{\infty}))\le \delta.
$$
Consider then the empirical sum
$$
\frac{1}{n}\sum_{i\in\np}f(y_i\sqrt{n\EE(\Delta_{\infty}^2)}).
$$
Writing 
$$
y_i=\frac{\beta^{-\ell}D_{\ell}(i)}{\sqrt{\sum_{j=1}^n \beta^{-2\ell}D_{\ell}^2(j)}},
$$
we have that this empirical sum differs from the simpler one
\begin{equation}\label{empirical_2} 
\frac{1}{n}\sum_{i\in\np}f(\beta^{-\ell}D_{\ell}(i))
\end{equation}
by at most 
$$
K\frac{1}{n}\sum_{i\in\np}\beta^{-\ell}|D_{\ell}(i)|\times \left|1-\sqrt{\frac{\EE(\Delta^2_{\infty})}{A}}\right|,
$$ 
where $K$ is the Lipschitz continuity constant of function $f$ and $A$ is the empirical sum in~(\ref{second_moments}). This correction tends to zero in probability by dominated convergence. Indeed, convergence to zero of $1-\sqrt{\EE(\Delta_{\infty}^2)/A}$ has just been established. By similar arguments as before based on Theorem~\ref{lemma_LD_bounds} and Lemmas~\ref{coupling_lemma_1} and~\ref{coupling_lemma_2}), the empirical average of the $|\beta^{-\ell}D_{\ell}(i)|$ is bounded. Convergence in probability of (\ref{empirical_2}) to $(1/2)\EE(f(\Delta_{\infty}))$ is then established by evaluating the first and second moments of this sum as previously done.

The same argument can be applied to $g$, eventually leading to the convergence in probability
$$
\lim_{n\to\infty}\frac{1}{n}\sum_{i\in\np}g(y_i\sqrt{n\EE(\Delta_{\infty}^2)})=\frac{1}{2}\EE(g(\Delta_{\infty})).
$$
It readily follows that
$$
\limsup_{n\to\infty}\left|\frac{1}{n}\sum_{i\in\np}\II_{y_i\ge \tau/\sqrt{n\EE(\Delta_{\infty}^2)}}-\frac{1}{2}\PP(\Delta_{\infty}\ge \tau)\right|\le \delta.
$$
As $\delta$ is arbitrary, this establishes~(\ref{empirical}).

Pick again an arbitrary $\delta>0$, two pairs of Lipschitz-continuous functions $f_{\pm}$ and $g_{\pm}$ such that 
$$
f_{\pm}(u)\le \II_{\pm u\ge t}\le g_{\pm}(u),
$$
and 
$$
0\le \EE(g_{\pm}(\pm\Delta_{\infty}))-\EE(f_{\pm}(\pm \Delta_{\infty}))\le \delta.
$$
The empirical overlap~(\ref{empirical_overlap}) reads
$$
\frac{1}{n}\left[\sum_{i\in \np}-\sum_{i\in\nm}\right]\left( 2\cdot \II_{x_i\sqrt{n\EE(\Delta_{\infty}^2)}\ge t}-1\right).
$$
The difference $(\np-\nm)/n$ is of order $1/\sqrt{n}$ and thus vanishes. We upper-bound the remaining terms by
$$
\frac{2}{n}\left[\sum_{i\in \np}g_+(x_i\sqrt{n\EE(\Delta_{\infty}^2)})-\sum_{i\in\nm}f_-(x_i(-\sqrt{n\EE(\Delta_{\infty}^2)})\right]
$$
Letting $K$ denote the Lipschitz-continuity constant for both $g_+$ and $f_-$, this last display differs from
\begin{equation}\label{empirical_3}
\frac{2}{n}\left[\sum_{i\in\np}g_+(\beta^{-\ell}D_{\ell}(i))-\sum_{i\in\nm}f_-(-\beta^{-\ell}D_{\ell}(i))\right]
\end{equation}
by at most 
\begin{equation}\label{empirical_4}
\frac{2K}{n}\sum_{i=1}^n\left|(x_i-y_i)\sqrt{n\EE(\Delta_{\infty}^2)}\right|+\beta^{-\ell}|D_{\ell}(i)|\times \left|1-\sqrt{\frac{\EE(\Delta^2_{\infty})}{A}}\right|.
\end{equation}
Because of the assumed convergence in probability $\lim_{n\to\infty}||x-y||=0$, the first error term necessarily tends to zero in probability by Cauchy-Schwarz inequality. The second term is dealt with as mentioned in the proof of the previous lemma. Finally, using the coupling lemmas~\ref{coupling_lemma_1} and~\ref{coupling_lemma_2}, by evaluating the first and second moments of (\ref{empirical_3}), we obtain the convergence in probability
$$
\lim_{n\to\infty}\frac{2}{n}\left[\sum_{i\in\np}g_+(\beta^{-\ell}D_{\ell}(i))-\sum_{i\in\nm}f_-(-\beta^{-\ell}D_{\ell}(i))\right]=\EE(g_+(\Delta_{\infty})-f_-(\Delta_{\infty})).
$$
The latter term is then an upper bound on the $\limsup$ of the empirical overlap. By the same approach, we obtain a lower bound of 
$$
\EE(f_+(\Delta_{\infty})-g_-(\Delta_{\infty}))
$$
on the $\liminf$ of the overlap. These upper and lower bounds differ by at most $2\delta$, and differ from $\PP(\Delta_{\infty}\ge t)-\PP(\Delta_{\infty}\le -t)$ by at most $\delta$. Since $\delta$ is arbitrary, this establishes the announced convergence in probability of the empirical overlap to quantity $x$ where
$$
x=\PP(\Delta_{\infty}\ge t)-\PP(\Delta_{\infty}\le -t)
$$
is strictly positive by our choice of $t$.
\end{proof}
\section{Proof of Lemma~\ref{lemma_align}}
\begin{proof}
The first and second evaluations follow by noting that the vectors whose difference is considered in the left-hand side agree on the set of entries $i$ whose $\ell$-neighborhood is cycle-free. Thus
$$
\begin{array}{ll}
|B^{(\ell)}e-\{S_{\ell}(i)\}|&\le \sqrt{|{\mathcal{B}}|}O(\log(n) \alpha^{\ell})\\
&\le O(\log^3(n)\alpha^{2\ell}),
\end{array}
$$
and the same bound holds for $|B^{(\ell)}\sigma-\{D_{\ell}(i)\}|$. This upper bound is $o(\sqrt{n}\beta^{\ell})$ so that the first two assertions follow, by further noticing that $|\{D_{\ell}(i)\}|=\Theta(\beta^{\ell})$, as follows from combining Theorem~\ref{lemma_LD_bounds} with Lemmas~\ref{coupling_lemma_1} and~\ref{coupling_lemma_2} as in the proof of Theorem~\ref{thm_local}. 

For the third assertion, consider the scalar product $<\{S_{\ell}(i)\},\{D_{\ell}(i)\}>$. Using the same arguments as in the proof of Theorem~\ref{thm_local} we obtain that this scalar product is indeed $o(|\{S_{\ell}(i)\}|\times |\{D_{\ell}(i)\}|$.
\end{proof}
\section{Proof of Lemma~\ref{lemma_trees}}
\begin{proof}
To establish the lower bound of~(\ref{norm_bound_1}), note that by Cauchy-Schwarz,
$$
<e, B^{(\ell)} B^{(\ell)}e>\le |e|\times |B^{(\ell)} B^{(\ell)}e|.
$$
However the left-hand side reads $|B^{(\ell)}e|^2$. Thus
$$
|B^{(\ell)} B^{(\ell)}e|\ge \frac{|B^{(\ell)}e|^2}{|e|}\cdot
$$
However it holds that $|B^{(\ell)}e|=\Theta(\sqrt{n}\alpha^{\ell})$ (this follows from the methods in the proof of Theorem~\ref{thm_local}). Since $|e|=\sqrt{n}$, the lower bound in~(\ref{norm_bound_1}) follows. For the upper bound, we note that by Lemma~\ref{det_grow2} and Theorem~\ref{lemma_LD_bounds}, the max row sum for matrix $B^{(\ell)}$ is of order $O(\log(n)\alpha^{\ell})$. 

The lower bound in (\ref{norm_bound_2}) is established similarly, from the inequality
$$
<\sigma, B^{(\ell)} B^{(\ell)}\sigma>\le |\sigma|\times |B^{(\ell)} B^{(\ell)}\sigma|
$$ 
The upper bound requires additional arguments. Assuming the $2\ell$-neighborhood of $i$ is cycle-free, the $i$-th entry of vector $B^{(\ell)} B^{(\ell)}\sigma$ can be written as
$$
\sum_{d=0}^{\ell}\sum_{j:d_G(i,j)=2d}\sigma_j |\{k: d_G(i,k)=d_G(j,k)=\ell\}|.
$$
We control the magnitude of this quantity in the tree model; using coupling we will then transpose the corresponding estimates to the original scenario.

Let then ${\mathcal{T}}$ denote a branching process with offspring $\hbox{Poi}(\alpha)$. The process of spins is then constructed by sampling uniformly the root's spin, and then propagating spins in a Markovian fashion with transition matrix $(a/(a+b) b(a+b), b(a+b), a(a+b))$ that is $\alpha^{-1}M$. Its eigenvalues are thus $(1,\beta/\alpha)$. 

The variable of interest is written
$$
X=\sum_{d=0}^{\ell} \sum_{j:d(j,i)=2d}\sigma_j |\{k: d(j,k)=d(i,k)=\ell\}|.
$$
We evaluate its second moment conditionally on ${\mathcal{T}}$ by writing $X^2$ as 
$$
\begin{array}{ll}
X^2=&\sum_{d=0}^{\ell}\sum_{d'=0}^{\ell}\sum_{j':d(j',i)=2d'}\sum_{j:d(j,i)=2d}\sigma_j \sigma_{j'}\times\\
&\times |\{k: d(j,k)=d(i,k)=\ell\}|\cdot |\{k': d(j',k')=d(i,k')=\ell\}|.
\end{array} 
$$
Now it holds that 
$$
\EE(\sigma_j \sigma_{j'}|{\mathcal{T}})=O\left(\left(\frac{\beta}{\alpha}\right)^{d(j,j')}\right).
$$
We will use this formula, and further distinguish nodes $j'$ according to their distance $2(d+d'-\tau)$ for $\tau=0,\ldots,2 (d\wedge d')$. This yields
$$
\begin{array}{ll}
\EE(X^2|{\mathcal{T}})=&\sum_{d,d'=0}^{\ell}\sum_{\tau=0}^{2(d\wedge d')}\sum_{j':d(j',i)=2d'}\sum_{j:d(j,i)=2d}O(\left(\left(\frac{\beta}{\alpha}\right)^{2(d+d'-\tau)}\right)\times\\
&\times |\{k: d(j,k)=d(i,k)=\ell\}|\cdot |\{k': d(j',k')=d(i,k')=\ell\}|.
\end{array}
$$
Note now that with high probability, we have the following evaluations
$$
\begin{array}{ll}
|\{k: d(j,k)=d(i,k)=\ell\}|&=\tilde{O}(\alpha^{\ell-d}),\\
|\{k': d(j',k')=d(i,k')=\ell\}|&=\tilde{O}(\alpha^{\ell-d'}),\\
|\{j:d(j,i)=2d\}|&=\tilde{O}(\alpha^{2d}),\\
|\{j':d(j',i)=2d'\&d(j,j')=2(d+d'-\tau)\}|&=\tilde{O}(\alpha^{2d'-\tau}).
\end{array}
$$
Plugging these in, we have 
$$
\begin{array}{ll}
\EE(X^2|{\mathcal{T}})&=\sum_{d,d'=0}^{\ell}\sum_{\tau=0}^{2(d\wedge d')}\tilde{O}\left(\left(\frac{\beta}{\alpha}\right)^{2(d+d'-\tau)}\right)
\alpha^{2\ell-d-d'+2(d+d')-\tau}\\
&=\sum_{d,d'=0}^{\ell}\sum_{\tau=0}^{2(d\wedge d')}\tilde{O}\left(\alpha^{2\ell}(\frac{\beta^2}{\alpha})^{d+d'-\tau}\right)\\
&=\tilde{O}(\alpha^{2\ell}(\beta^2/\alpha)^{2\ell})\\
&=\tilde{O}(\beta^{4\ell}).
\end{array}
$$
By coupling (techniques of Theorem~\ref{thm_local} involving Tchebitchev inequality, based on the bounds of Theorem~\ref{lemma_LD_bounds} and Lemmas~\ref{coupling_lemma_1} and~\ref{coupling_lemma_2}) we thus have that with high probability,
$$
|B^{(\ell)}B^{(\ell)}\sigma|=\tilde{O}(\sqrt{n\beta^{4\ell}})=\tilde{O}(\beta^{\ell}|B^{(\ell)}\sigma|)
$$
as announced.
\end{proof}
\end{appendix}
\end{document}